\documentclass[10pt,conference]{IEEEtran}
\usepackage{amsmath}
\usepackage[utf8]{inputenc}
\usepackage{pdfpages}
\usepackage{graphicx}
\usepackage{amsmath}
\usepackage{epstopdf}
\usepackage{amssymb}
\usepackage{mathtools}
\usepackage{mathtools}
\usepackage{amsmath}
\usepackage{amstext}
\usepackage{amssymb}
\usepackage{amsfonts}
\usepackage{bm}
\usepackage{algorithm}
\usepackage{algorithmic}

\usepackage{tabularx}
\usepackage{latexsym}
\usepackage{paralist}
\usepackage{float}
\usepackage{enumerate}

\usepackage{paralist}
\usepackage{amsopn}
\usepackage{mathrsfs}
\usepackage{mathtools}
\usepackage{amsthm}
\usepackage{amssymb}
\usepackage{amsfonts}
\usepackage{comment}
\usepackage{xspace} 
\usepackage{relsize}
\usepackage{float}
\usepackage{bbm}

\theoremstyle{plain}
 
\newtheorem{lemma}{Lemma}

\newtheorem{cor}{Corollary}

\theoremstyle{remark}

\newcommand{\Files}{\mathcal{F}}

\newcommand{\BS}{\mathcal{M}}
\newcommand{\bs}{m}
\newcommand{\numBS}{M}

\newcommand{\numUT}{N}
\newcommand{\R}{\mathbb{R}}

\newcommand{\fa}{\hspace{5pt}\forall\hspace{2pt}}

\newcommand{\Regions}{\mathcal{S}}
\newcommand{\region}{s}

\newcommand{\RegFile}{\mathcal{Q}}

\newcommand{\numRegFiles}{Q}
\newcommand{\crpConst}{a}

\newcommand{\norm}[1]{\left\lVert#1\right\rVert}

\DeclareMathOperator*{\argmax}{arg\,max}

\DeclareMathOperator{\coverBS}{\BS}

\DeclareMathOperator{\util}{U}

\begin{document}

\title{Fair distributed user-traffic association in cache equipped cellular networks}


\author{\IEEEauthorblockN{Jonatan Krolikowski\IEEEauthorrefmark{1}, Anastasios Giovanidis\IEEEauthorrefmark{2} and Marco Di Renzo\IEEEauthorrefmark{1}} 
\IEEEauthorblockA{\IEEEauthorrefmark{1}CNRS-L2S, CentraleSup{\'e}lec, Université Paris Sud, Universit{\'e} Paris-Saclay\\
\IEEEauthorrefmark{2}CNRS-LIP6, Universit{\'e} Pierre et Marie Curie, Sorbonne Universit{\'e}s, Paris, France\\
Email: jonatan.krolikowski@telecom-paristech.fr, anastasios.giovanidis@lip6.fr, marco.direnzo@lss.supelec.fr}}

\maketitle

\date{\today}

\maketitle
\begin{abstract}
Caching of popular content on wireless nodes is recently proposed as a means to reduce congestion in the backbone of cellular networks and to improve Quality of Service. From a network point of view, the goal is to offload as many users as possible from the backbone network to  the wireless caches while at the same time offering good service to cache-unrelated users. Aggressive offloading  can lead to  an unbalanced user association. Some wireless nodes  can be overloaded by cache-related traffic while   the resources of  others remain underused.

Given a fixed content placement, this work proposes an efficient distributed  algorithm to control and balance the association of cache-related traffic among  cellular cache memories. The algorithm allows the network to achieve the globally optimal solution and can be  executed  on  base stations using a limited amount of information exchange between them. It is based on a novel algorithm we call \emph{Bucket-filling}. The solution limits the cache-users per node by  balancing the total load among the nodes in a fair way. 
The improvement compared to common user assignment policies is highlighted    for  single- as well as for multi-tier random networks. 
\end{abstract}

\section{Introduction}
 
By 2020, wireless data traffic is estimated to reach roughly the 8-fold of its volume of 2015~\cite{Cisco2016}. 
Such increase in data demand will be satisfied by densifying the network with new tiers
as well as by allowing cooperation among  stations. 
However, this increase of wireless traffic can pose new problems to the wireless backhaul network that are related to congestion. 
A trending strategy to ease the backhaul is to equip wireless nodes with large cheap cache memories \cite{poularakis2016}. Content can be stored on them based on some  knowledge of the popularity of the content requested by the users in the cellular network. The content  is thereby brought closer to the user and the costly usage of backhaul bandwidth is reduced. 

 
Since not all requests can be served by the caches, wireless traffic will be divided into (1) cache-related and (2) cache-unrelated traffic. The former can retrieve the requested content from the cache of a node whose signal can be sufficiently received. The latter will typically connect to the node with the strongest signal. As a result, wireless nodes with popular content could be overloaded with cache-related traffic, leaving no resources for cache-unrelated users. At the same time, other nodes with less popular content could be underused.

A solution for such imbalance can be to limit the maximum cache-related load per wireless node and redistribute, when possible, the remaining load among neighboring nodes. This way, resources can be kept available for cache-unrelated traffic which is less predictable. Note that although multicasting of cached content could reduce the effective traffic load, this approach also has its limitations due to service delay tolerance \cite{poularakis2016a}. Thus, the solution of the above association problem remains necessary.

 
The existing literature has not adequately dealt with such questions yet. It focuses on how to determine good content placement policies with respect to certain optimality criteria such as the hit ratio or  user delays.  User association is handled in different ways: 
In \cite{Borst2010}, the authors use caching to minimize bandwidth cost in a tree-like network. The routing decisions of  users are, however, independent of each other. 
The authors of~\cite{Caire2013} associate users to any covering station that caches their requested content without balancing the traffic loads. 
 In other works~\cite{Blaszczyszyn2014}\cite{Bastug2015}, users are associated to the closest base station, not knowing if the requested content is stored in the cache or not. 
The authors of \cite{Poularakis2014b} maximize the hit ratio by means of integer optimization. They introduce a bandwidth constraint limiting   the amount of users that can be connected to each cellular station.  While this model avoids  overloading  cellular stations, it cannot guarantee a balanced routing among them. 
In \cite{Dehghan2015},  user association is balanced between a cached and an uncached path. Association to the individual caches is modeled by shortest distance, again not allowing control over the use of the separate resources.
The model in~\cite{Naveen2015} includes both fractional content placement and routing variables and allows for the balancing of user traffic loads at the cache-equipped base stations. In the  solution of the problem, however, the convergence to an optimal routing is dependent on iterative fractional content placement updates, which is decided for  the entire network.  

In this paper, we model the problem of associating users to stations with given cached content by introducing a utility function per station which puts soft limitation on the served user load.
The related optimization problem  allows to enforce a load fairness criterion between caches. We further develop a policy to optimally solve this problem in a distributed way.  For a given cache placement, the resulting policy guarantees that all resources are used as evenly as possible.  The calculations can be executed 
 on the individual stations  requiring a limited amount of information exchange. 
 We show that  the policy is beneficial both in single- and in multi-tier networks.



The remainder of this paper is organized as follows: 
Section~\ref{section:model} presents  the network model and the  problem formulation. In Section~\ref{section:solution}, we introduce the solution techniques: Augmented Lagrangian, Diagonal Quadratic Approximation and the novel Bucket-filling algorithm. 
 Numerical evaluations of the resulting  policy  are presented in Section~\ref{section:numerical} for random networks (single-tier and 2-tier) with varying coverage or node density. Finally, Section~\ref{section:conclusion} concludes our work.

\section{System Model and Problem statement}

\label{section:model}

Consider the downlink of a cellular communications network with a finite set $\BS = \lbrace \bs_1, \ldots, \bs_{\numBS}\rbrace$ of Cached Base Stations (CBSs).  
 Every CBS is equipped with a  cache memory 
  in which content files from a finite catalog 
  $\Files$    
  are stored. 
 Note that other base stations not equipped with caches (UBSs) can be part of the network (not considered in  $\BS$). Each CBS has a coverage area, and areas of different CBSs can overlap. 
 
A user arrives in  a position on the plane  and  requests  a file   $f\in\Files$.  
 The user will be covered (or not) by some CBS. In case he/she is covered by some CBS with the requested content in the cache, the user will be served from the caches. This traffic is called \emph{cache-related}. The user association is assumed unique in the sense that an association to two or more CBSs is not allowed. When a user request occurs in the overlap of two or more coverage areas, it can be associated to any one of the covering CBSs. The exact distance of the user from the covering station (e.g.\ near the center or at the edge of the cell) is not taken into consideration. 
The users who do not find their requested content cached, and the users not covered by any CBS at all, constitute the \emph{cache-unrelated} traffic. This traffic will be served either from covering CBSs or UBSs. In this work, we will not treat the question how to distribute its load among the stations.
The boundaries of the overlapping coverage areas of the CBSs partition 
the plane into a set of disjoint regions: Two points 
are in the same region if and only if they are in the coverage areas of exactly the same subset  of CBSs. Let $\Regions$ denote the set of all regions 
 covered by at least one CBS. 

The cache placement is given apriori. 
Tuples of region $\region$ and file $f$ are called \emph{region-files} $q=(\region, f)$ if there is a CBS covering $\region$ with $f$ in its cache. 
We   denote the complete set of region-files by
 $\RegFile$. Not included in  $\RegFile$ are requests for files not stored in a covering CBS. 
The subset of CBSs that cache content $f$ and cover region $\region$  is denoted by $\coverBS(q)\subseteq \BS$. 
Conversely, $\RegFile(\bs)\subseteq\RegFile$ are the region-files which can be served by the cache of $\bs$.


For delay and backhaul congestion reasons, it is always more beneficial to serve a request from a cache than through the backhaul network.   Therefore, in an optimal solution of the routing problem, all traffic of $q\in\RegFile$ is served from the cache of a CBS: For each region-file $q=(\region, f)\in\RegFile$ there is a station covering $\region$ with $f$ in its cache. 
  As a consequence, all  user requests  realted to region-files in $\RegFile$ (and only they) are cache-related traffic.

For each region-file $q\in\RegFile$, the expected number of user requests is denoted by $\numUT_{q}\in\R_{+}$. The vector of these popularity values is denoted by $\mathbf{\numUT} \coloneqq (\numUT_{q}), q\in\RegFile$. 
The information on file popularity is considered locally available at each covering CBS. Note that the model does not assume spatially uniform traffic and that the vector $\mathbf{\numUT}$ is general.
  
We introduce the routing variable  ${y_{\bs,q}\geq 0}$  for all ${\bs\in\BS}, {q\in\RegFile(\bs)}$. Each $y_{\bs,q}$ with $q=(\region,f)$ denotes the part  of expected traffic for content file $f$  in region $\region$   associated to CBS $\bs$.
The vector of all routing variables is denoted by $\mathbf{y}$.
For each $q\in\RegFile$, the sum of all routing variables towards the covering CBSs must be smaller than or  equal to $\numUT_{q}$. 
Formally, for every feasible solution:
 \begin{align}
 \label{CRP:routing}
\sum_{\bs\in \coverBS(q) }y_{\bs,q}  \leq \numUT_{q}, \fa q \in \RegFile.
\end{align}





  The total traffic volume associated with CBS $\bs$ depends on the partial vector $\mathbf{y}_{\bs}\coloneqq (y_{\bs,q}),q\in\RegFile(m)$ 
 and can be expressed as follows:
\begin{align*}
v_{\bs}(\mathbf{y}_{\bs}) \coloneqq \sum_{q\in\RegFile(\bs) }  y_{\bs,q}.
\end{align*}


The aim of our model is to put a soft limit on the maximum traffic associated with a CBS while balancing the cache-related traffic volume among the CBSs. This approach also guarantees the \emph{usefulness} of each cache.  We measure the usefulness of a CBS $\bs$ by the utility function   $\util_{\bs}$ which takes  the traffic volume routed to the CBS as its argument.

 Following the concept of diminishing returns, $\util_{\bs}$  is  increasing and concave as well as  continuously differentiable. Additionally, we can choose the utility function such that for volumes greater than a certain $V_{\bs}$, the derivative of $\util_{\bs}$ becomes close to zero. Such soft limitation entails that  it is not beneficial to further route users to this station. The overall objective is to maximize the sum of utilities subject to routing constraints. 
  The Convex Program  formulation of the Cache Routing Problem is 
 \begin{align*}
(\text{CRP})	&\max_{0\leq  \mathbf{y}_{\bs} \leq \mathbf{N}_{\bs},  \bs\in\BS } && \sum_{\bs\in\BS} \util_{\bs}( v_{\bs}(\mathbf{y}_{m}) )&\nonumber\\
	&\text{s.\ t.} && \sum_{\bs\in \coverBS(q)}y_{\bs,q} = \numUT_{q}, &\fa q\in\RegFile,
\end{align*}
  where $\mathbf{N}_{\bs} \coloneqq (\numUT_{\bs,q}), q\in\RegFile(\bs)$ is the vector of expected requests for the region-files covered by $\bs$. Note that we  introduce an equality constraint instead of the inequality  in \eqref{CRP:routing}, since the utility functions are increasing. There is, thus, no cache-related traffic  in the optimal solution that remains unrouted.

Traffic association is balanced  if the available resources are used in a fair way. Some notions of fairness are max-min, $\alpha-$ and proportional fairness. Each of them is achieved by appropriate choice of the utility functions (see~\cite{Kelly1997}\cite{ Mo2000}). E.g.\ for proportional fairness, utilities could be chosen as (weighted) logarithms, depending also on the soft limit we want to achieve.

In the next section, we  derive the procedure that achieves the solution to the CRP for general utilities, leaving their specific choice to the network designer.

\section{Solution}
\label{section:solution}

The CRP is solved with a distributed algorithm which consists of three nested loops. 

\subsection{Dual method for the Augmented Lagrangian}

The CRP is solved using the dual method  on the Augmented Lagrangian (see~\cite{Bertsekas1989}, Section 3.4.4). We use the Augmented instead of the regular Lagrangian  to achieve a distributed solution. In our case, the regular Lagrangian  is not appropriate since it is not strictly concave in the primal variables and hence the primal solution is not unique. This creates conflicts when different stations compete for the same users and convergence cannot be guaranteed.
Like the regular Lagrangian, the Augmented one relaxes  constraints of the  CRP and  introduces a \emph{price} $\lambda_{q}$ for the violation of each constraint. The difference between them is an additional quadratic term penalizing the violation of each constraint together with a factor $\varrho>0$. This penalty guarantees  strict concavity in the primal variables. 
Denoting the Augmented Lagrangian by $L^{(\varrho)}$, we get
\begin{align}
L^{(\varrho)}(\mathbf{y},\bm{\lambda}) = &\sum_{\bs\in\BS} \util_{\bs}( v_{\bs}(\mathbf{y}_{\bs}) ) \nonumber\\ &-\sum_{q\in\RegFile}  \lambda_{q}(  \numUT_{q} - \sum_{\bs\in \coverBS(q)}y_{\bs,q} ) \nonumber\\\label{CRP-AugLag}&- \frac{\varrho }{2}\sum_{q\in\RegFile}(\numUT_{q}-\sum_{\bs\in\coverBS(q)} y_{\bs,q})^2,
\end{align}
where $\bm{\lambda}\coloneqq(\lambda_{q}), q\in\RegFile$ is the price vector.
The domains of the dual variables  are $\lambda_{q}\in\R$ for all $q\in\RegFile$, since the respective constraints are equalities. 

The Duality theorem (see~\cite{Bertsekas1989}, Appendix C) applies, which means that the duality gap is 0, and the dual method can be used.  
The objective function of the dual problem  is 
\begin{align*}
D^{(\varrho)}(\bm{\lambda}) &\coloneqq \max_{0\leq  \mathbf{y}_{\bs} \leq \mathbf{N}_{\bs}, \bs\in\BS } L^{(\varrho)}(\mathbf{y},\bm{\lambda}) = L^{(\varrho)}(\mathbf{y^{\ast}}(\bm{\lambda}), \bm{\lambda}),
\end{align*}
where  
\begin{align}
\label{CRP-primal}
&&\mathbf{y^{\ast}}(\bm{\lambda}) = \argmax_{0\leq  \mathbf{y}_{\bs} \leq \mathbf{N}_{\bs}, \bs\in\BS } L^{(\varrho)}(\mathbf{y},\bm{\lambda})
\end{align}
is the primal maximum of \eqref{CRP-AugLag} for a given price vector $\bm{\lambda}$. 
The dual problem is then defined as
\begin{align*}
(\text{CRP-dual})	&&\min_{\bm{\lambda}\in\R^{\RegFile}} D_{\varrho}(\bm{\lambda}).
\end{align*}

Starting from an arbitrary initial dual vector $\bm{\lambda}(0)$, the dual vector is iteratively updated  according to
\begin{align}
\lambda_{q}(t+1) = \lambda_{q}(t) + \varrho\Big(\sum_{\bs\in\coverBS(q)} \numUT_{q} -  y_{\bs, q}^{\ast}(\bm{\lambda}(t)) \Big),
\label{crp:dual-update}
\end{align}
where the steplength $\varrho>0$ is the penalty used in \eqref{CRP-AugLag}.
The convergence of this method is well known (see~\cite{Ruszczynski1995} or Section 3.4.4 of \cite{Bertsekas1989}). 

For practical implemetation issues of each update step of the region-file price $\lambda_{q}, q\in\RegFile$, only the primal solutions of the covering CBSs need to be known. Thus, for a distributed implementation, exchange of such  information among neighboring stations is  sufficient.

The next subsection presents the distributed solution for the primal problem \eqref{CRP-primal}, which needs to be found for every iteration of the dual algorithm.


\subsection{Distributed solution for the primal problem}

The solution for  \eqref{CRP-primal}  is unique since the domain of  $\mathbf{y}$ is convex and compact and,  for any fixed feasible vector $\bm{\lambda}$, the Augmented Lagrangian $L^{(\varrho)}$ is strictly concave. 
We use the Diagonal Quadratic Approximation Method (DQA)~\cite{Ruszczynski1995} to derive seperate problems which can be solved by each cache. A limited amount of exchanged information  between neighboring caches is required. 

The DQA overcomes the problem that the objective function $L^{(\varrho)}(\mathbf{y},\bm{\lambda})$   of~\eqref{CRP-primal} is not easily separable among the variables related to the different CBSs, since it contains  quadratic terms combining different variables $y_{\bs, q}$ (see \eqref{CRP-AugLag}). 
To achieve this, we  introduce the functions $L^{(\varrho)}_{\bs}:\R^{\numRegFiles_{\bs}}\times\R^{\sum_{\tilde{\bs}}\numRegFiles_{\tilde{\bs}}}\times \R^{\numRegFiles}\rightarrow\R$ for all $\bs\in\BS$:
\begin{align*}
L^{(\varrho)}_{\bs}(\mathbf{y}_{\bs},\mathbf{\tilde{y}},\bm{\lambda}) \coloneqq & \util_{\bs}( v_{\bs}(\mathbf{y}_{\bs}) )  + \sum_{q\in\RegFile(m)}  \lambda_{q} y_{\bs,q}  \nonumber\\&- \frac{\varrho }{2}\sum_{q\in\RegFile(m)}(\bar{\numUT}^{m}_{q}(\mathbf{\tilde{y}})- y_{\bs,q})^2,
\end{align*}
where 
$
\bar{\numUT}^{m}_{q}(\mathbf{\tilde{y}}) \coloneqq \numUT_{q}  -  \sum_{\substack{\bar{\bs}\in\coverBS(q)\\ \bar{\bs} \neq \bs}}  \tilde{y}_{\bar{\bs}, q}
$
is the number of requests in $q$ not associated with caches other than $\bs$ in the  routing vector $\mathbf{\tilde{y}}=(\tilde{y}_{\bs, q}), \bs\in\BS, q\in\RegFile(\bs)$ which is here seen as a  parameter.
The primal problem to be solved by each cache $\bs$ is defined as
\begin{align}
\label{CRP:primal-sep-def}
(\text{CRP-primal-$\bs$})\quad \max_{0\leq \mathbf{y}_{m} \leq \mathbf{N}_{m}} L^{(\varrho)}_{\bs}(\mathbf{y}_{\bs},\mathbf{\tilde{y}},\bm{\lambda}).
\end{align} 

Since $L^{(\varrho)}_{\bs}(\mathbf{y}_{\bs},\mathbf{\tilde{y}},\bm{\lambda})$ is strictly concave in $\mathbf{y}$ and the domain is compact, CRP-primal-$\bs$ has a unique solution which we call $\mathbf{\tilde{y}^{\ast}}_{\bs}$. The vector containing  the solutions of CRP-primal-$\bs$ for all caches is $\mathbf{\tilde{y}^{\ast}}$. 

The DQA method consists of parallel execution of CRP-primal-$\bs$ at the caches with consecutive update of the vector $\mathbf{\tilde{y}}$ in the fashion of a nonlinear Jacobi algorithm. It  produces a succession of vectors $\mathbf{\tilde{y}}(0), \mathbf{\tilde{y}}(1), \mathbf{\tilde{y}}(2),\ldots$. Starting from some given vector $\mathbf{\tilde{y}}(0)$, the vector  $\mathbf{\tilde{y}}(\tau + 1)$ is defined as the convex combination of $\mathbf{\tilde{y}}(\tau )$ and $\mathbf{\tilde{y}^{\ast}}(\tau)$. Given a constant $0<\alpha\leq 1$, we get
\begin{align}
\mathbf{\tilde{y}}(\tau + 1) = \mathbf{\tilde{y}}(\tau) + \alpha(\mathbf{\tilde{y}^{\ast}}(\tau)-\mathbf{\tilde{y}}(\tau)).
\label{crp:primal-update}
\end{align}
In~\cite{Ruszczynski1995} it is shown that the DQA method converges. Observe that the convergence depends on the uniqueness of the primal solutions $\mathbf{\tilde{y}^{\ast}}(\tau)$. 
For every update~\eqref{crp:primal-update}, each station only requires results from its neighboring stations that cover a common region-file. 

\subsection{Bucket-filling}

What  is left  is to find an efficient solution to \mbox{CRP-primal-$\bs$}~\eqref{CRP:primal-sep-def} running on each CBS separately:
In order to do so, we develop  a novel technique we call \emph{Bucket-filling}. First, we can  simplify our notation 
since the problem is separated by cache, and the vectors $\mathbf{\tilde{y}}$ and $\bm{\lambda}$ are parameters. In this subsection, ${\mathbf{y} \coloneqq \mathbf{y}_{m}}$, ${\util\coloneqq \util_{\bs}}$, ${\numUT\coloneqq \numUT_{m}}$, ${v\coloneqq v_{\bs}}$, $\RegFile\coloneqq\RegFile(\bs)$, $\bar{\numUT}_{q}\coloneqq\bar{\numUT}^{m}_{q}(\mathbf{\tilde{y}})$.
Then, the solution to \eqref{CRP:primal-sep-def} can be rewritten as 
\begin{align*}
\mathbf{y^{\ast}}=&\,\argmax_{0\leq \mathbf{y} \leq \mathbf{N}} \, \util( v(\mathbf{y}) )  + \sum_{q\in\RegFile}  \lambda_{q} y_{q}   - \frac{\varrho }{2}\sum_{q\in\RegFile}(\bar{\numUT}_{q}- y_{q})^2\\
=&\,\argmax_{0\leq \mathbf{y} \leq \mathbf{N}} \, \util\Big( \sum_{q\in\RegFile} y_{q} \Big) - \sum_{q\in\RegFile} \left[  \frac{\varrho }{2} y_{q}^{2} -  (\lambda_{q} + \varrho \bar{\numUT}_{q}) y_{q}\right].
\end{align*}
The last step comes from the development of the quadratic term and from omitting the  additive constants which do not affect  the optimal solution.
Further defining  $\crpConst_{q} \coloneqq \lambda_{q} + \varrho \bar{\numUT}_{q}$, the CRP-primal-$\bs$ can be stated as 
\begin{align}
\label{CRP:prim-approx-reform}
\mathbf{y^{\ast}} = \,\argmax_{0\leq \mathbf{y} \leq \mathbf{N}} \quad g(\mathbf{y})
\end{align}
with
$
g(\mathbf{y})\coloneqq \util\Big( \sum_{q\in\RegFile} y_{q} \Big)  - \sum_{q\in\RegFile} \left[  \frac{\varrho }{2}  y_{q}^{2} -  \crpConst_{q} y_{q}\right].
$

We will now show that one known value in the optimal routing vector is sufficient to derive the entire solution of CPR-primal-$\bs$. This observation allows to develop  the Bucket-filling technique.

 Given the  optimal vector $\mathbf{y^{\ast}}$, we define the function $g_{q}(y):\left[0,\numUT_{q}\right]\rightarrow\R$ for  $q\in\RegFile$ as
\begin{align*}
g_{q}(y)\coloneqq & \util\Big(y + \sum_{\tilde{q}\neq q} y_{\tilde{q}}^{\ast} \Big) - \left[  \frac{\varrho }{2}  y^{2} -  \crpConst_{q} y\right] \nonumber\\&- \sum_{\tilde{q}\neq q} \left[  \frac{\varrho }{2} ( y_{\tilde{q}}^{\ast})^{2} -  \crpConst_{\tilde{q}} y_{\tilde{q}}^{\ast}\right].
\end{align*}
The maximum of this function is  the optimal solution $y^{\ast}_{q}$. 
In other words,
$
\argmax_{0\leq y \leq \numUT_{q}} \, g_{q}(y) = y_{q}^{\ast}.
$

Furthermore, note that $g_{q}$ is continuously differentiable, strictly concave and defined on a compact set which implies that it has a unique maximum. We  denote the unique root of its  derivative
\begin{align}
\label{CRP:prim-approx-deriv}
g^{\prime}_{q}(y) = \util^{\prime}\Big(y + \sum_{\tilde{q}\neq q} y_{\tilde{q}}^{\ast} \Big) - \left[  \varrho   y -  \crpConst_{q} \right]
\end{align}
by $y^{\prime}_{q}$.
Then we get 
\begin{align}
0\leq y^{\prime}_{q}\leq \numUT_{q} &\implies y^{\ast}_{q} = y^{\prime}_{q},\nonumber\\
y^{\prime}_{q} < 0 &\implies y^{\ast}_{q} = 0, \label{CRP:prim-approx-obs}\\
y^{\prime}_{q} > \numUT_{q} &\implies y^{\ast}_{q} = \numUT_{q}.\nonumber
\end{align}


Now, we can formulate the following key lemma:
\begin{lemma}
\label{crp-main-lemma}
Let $\mathbf{y^{\ast}}$ be defined as in~\eqref{CRP:prim-approx-reform}. 
Let $i,j\in\RegFile$ with  
$0<y^{\ast}_{i}<\numUT_{i}$, $0<y^{\ast}_{j}<\numUT_{j}$. Then
\begin{align*}
y_{i}^{\ast} =  \frac{\crpConst_{i}-\crpConst_{j}}{\varrho} + y_{j}^{\ast}.
\end{align*}
\end{lemma}
\begin{proof}
From \eqref{CRP:prim-approx-obs} we know that $g^{\prime}_{i}(y_{i}^{\ast}) = 0$ and $g^{\prime}_{j}(y_{j}^{\ast}) = 0$ and thus $g^{\prime}_{i}(y_{i}^{\ast}) = g^{\prime}_{j}(y_{j}^{\ast}) $. 
From~\eqref{CRP:prim-approx-deriv}, we get
\begin{align*}
\varrho   y_{i}^{\ast} -  \crpConst_{i} =    \varrho   y_{j}^{\ast} -  \crpConst_{j} 
\end{align*}
The statement follows from a simple calculation.
\end{proof}

The Lemma allows us to characterize the entire optimal solution just from one  value 
$0<y^{\ast}_{q}<\numUT_{q}$.

\begin{cor}
\label{obs-next-increase}
The steepest increase of $g$  (as in \eqref{CRP:prim-approx-reform}) is achieved by increasing  the region-file $q$ with the highest value $\crpConst_{q}$.
For two different $q, \hat{q}\in\RegFile$, the gradients of $g_{q}$ and $g_{\hat{q}}$ become equal once $y_{q}$ is increased by the value $(\crpConst_{q}-\crpConst_{\hat{q}})/\varrho$. 
\end{cor}

We can now describe the novel Bucket-filling algorithm that efficiently finds the optimal solution to CRP-primal-$\bs$: 
\begin{algorithm}[H]
\caption{Bucket-filling}
\begin{algorithmic}[1]
\STATE Sort $\RegFile$ by $\crpConst_{q}$ non-increasingly. \label{algo-sort}
\STATE Declare all $q\in\RegFile$ with maximum $\crpConst_{q}$ as \emph{active}.\label{algo-activate} 
\STATE \label{algo-increase} Increase $y_{q}$ of active $q$ equally (because of Lemma~\ref{crp-main-lemma}) until
\begin{compactitem}
\item another $q$ becomes active by Corollary~\ref{obs-next-increase},
\item a $q$ becomes inactive by $y_{q}$ reaching $\numUT_{q}$, or
\item  $g^{\prime}(\mathbf{y})=0$ for all active region-files $q$.
\end{compactitem} 
\STATE If the last condition is fulfilled, or  $y_{q} = \numUT_{q}$ for all $q$, then terminate. Otherwise, go to step~\ref{algo-increase}. \label{algo-terminate} 
\end{algorithmic}
\label{algo:bucket-filling}
\end{algorithm}


The algorithm is illustrated in Fig.~\ref{fig:crp}. It terminates and returns the optimal solution $\mathbf{y^{\ast}}$ defined in \eqref{CRP:prim-approx-reform}.

\begin{figure}[h]
\centering
\includegraphics[width=\columnwidth]{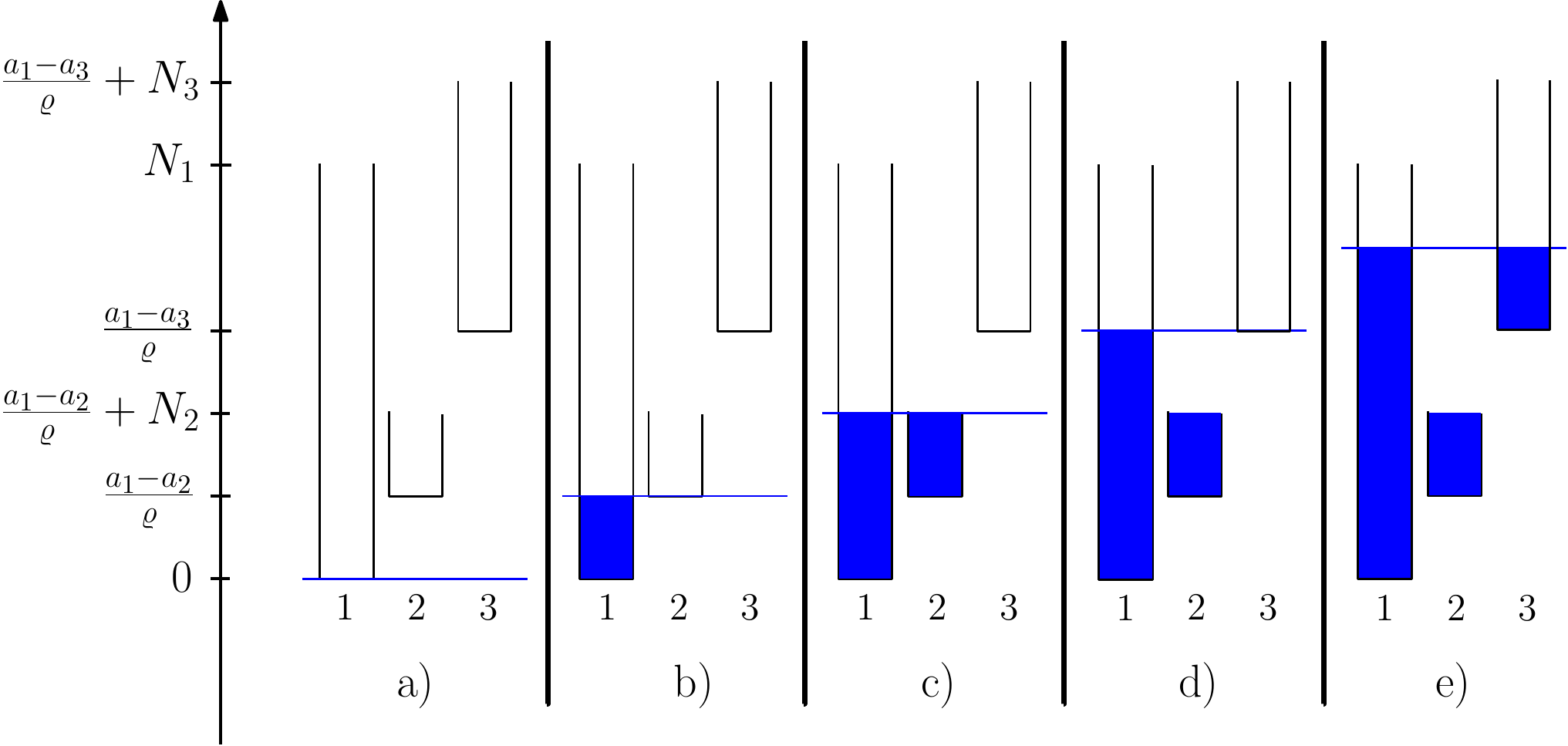}
\caption{Illustration of the Bucket-filling algorithm. Each bucket $q$ is placed with its bottom at level $\frac{\crpConst_{\max}-\crpConst_{q}}{\varrho}$ and has the height $\numUT_{q}$. The buckets are filled to a common level -- or until they are full. The algorithm halts when further increasing the water levels starts decreasing the objective value.}
\label{fig:crp}
\end{figure}

Observe that step~\ref{algo-sort} can be done in $O(\left\vert\RegFile\right\vert \log(\left\vert\RegFile\right\vert))$ operations with a sorting algorithm such as quicksort. Step~\ref{algo-activate} runs in  $O(\left\vert\RegFile\right\vert)$. Step~\ref{algo-increase} is executed up to $2\left\vert\RegFile\right\vert$ times, since every region-file is activated and deactivated no more than one time each. Assuming that $f$ can be evaluated in $O(1)$, this implies $O(\left\vert\RegFile\right\vert)$ as the asymptotic runtime for steps~\ref{algo-increase} and \ref{algo-terminate}. This shows that the  overall runtime is dominated by sorting in step~\ref{algo-sort} and thus is  $O(\left\vert\RegFile\right\vert \log(\left\vert\RegFile\right\vert))$.


This algorithm is executed in every inner loop of the DQA method, which is why its efficiency is paramount.

\subsection{Algorithm}
\label{section:algorithm}

To sum up, the complete algorithm that findes the globally optimal solution to the CRP  is
\begin{algorithm}[H]
\caption{Solve CRP}
\begin{algorithmic}[1]
\STATE Choose dual vector~$\bm{\lambda}(0)$, $t=0$, $\varepsilon > 0$. \label{algo2:choose-dual}
\WHILE{$\norm{{\bm{\lambda}(t)} - {\bm{\lambda}(t-1)}} > \varepsilon$ }
	\STATE Choose 	$\mathbf{\tilde{y}}(0)$, $\tau = 0$. \label{algo2:choose-primal}
	\WHILE{$\norm{\mathbf{\tilde{y}}(\tau) - \mathbf{\tilde{y}}(\tau-1)} > \varepsilon$ }
			\STATE Find $\mathbf{\tilde{y}^{\ast}}_{\bs}(\tau)$ with Algorithm~\ref{algo:bucket-filling} at every $\bs\in\BS$ separately.
			\STATE Exchange results among neighboring stations, set $\mathbf{\tilde{y}}(\tau + 1)$ as in \eqref{crp:primal-update}, $\tau = \tau + 1$.
	\ENDWHILE
	\STATE Exchange results among neighboring stations, set ${\bm{\lambda}(t+1)}$ as in \eqref{crp:dual-update} , $t = t + 1$.
\ENDWHILE
\end{algorithmic}
\label{CRP:algo}
\end{algorithm}
The choice of the first dual vector~$\bm{\lambda}(0)$ in line~\ref{algo2:choose-dual} is arbitrary. The first primal vector~$\mathbf{\tilde{y}}(0)$ in line~\ref{algo2:choose-primal} can be chosen as the last primal vector of the iteration before.

\section{Numerical Evaluation}
\label{section:numerical}

For the simulation of the algorithm, we consider an urban area of $2.5$~km $\times$ $2.5$~ km with uniform user distribution. A catalog of $6$ files of equal size is known. The popularity of the files  follows a Zipf distribution with parameter $1$. Throughout the simulations,  CBSs are placed in the area following a Poisson Point Process (PPP) with  density  $8 \frac{\text{CBS}}{km^{2}}$. This means that their total number in each run is a random Poisson realization, and their position is uniform in the simulation window. In this way we avoid the bias of testing only particular network topologies. We run two simulation scenarios, the first for single-tier and the second for two-tier netowrks. Each scenario consists of 1000 simulation runs  and we consider the averaged results over the runs. Coverage follows the Boolean model where a disc area is centered on each wireless station with some defined radius. The surface of different overlapping areas is found in each run by the Monte-Carlo method. The users are routed to the CBSs following three different policies:
\begin{compactenum}
\item The policy from Algorithm~\ref{CRP:algo} with logarithmic utility function for each cache. This policy guarantees a proportionally fair (and also max-min fair) solution. We call this policy \emph{fair}.
\item The  \emph{closest-available} policy, which associates each user with the closest CBS  that both covers its position and has the requested content cached.
\item The \emph{unsplittable} policy which associates all users in a region requesting the same file with a unique random CBS among all covering CBSs having this content.
\end{compactenum}
We want to evaluate the proportion of user traffic served by different  CBSs over the total traffic routed to  CBSs for each policy. In this way we can  compare the policies based on how (un)equally they associate traffic load among the available CBSs.
Observe that for all three policies, the total traffic volume associated  to CBSs is the same, because in all scenarios the CBSs store the same cached content, and  traffic is  routed to a CBS whenever possible. Hence, the comparison is fair. 

\subsection{Single-tier Networks}
In an ideal situation, all stations would serve exactly the same amount of traffic. This, however, is normally not possible in a random network. A routing policy is better than another, when the maximum load of a CBS is lower and at the same time the minimum load share is higher than in the other policy. This way, an overload of the stations is avoided while the usefulness of all stations is achieved. 
Simulating single-tier networks, we want to verify that the fair policy provides a more balanced distribution of traffic to CBSs than the other two. The coverage radius of the CBSs is varied between 62.5 m and 500 m. This can be translated to an expected number of covering CBSs per user between 1 and 6. This mapping comes from the Boolean model \cite{Blaszczyszyn2014}. Two different sets of content files with different popularities are placed uniformly randomly into the caches. Since we are only interested in traffic associated with  CBSs, we disregard users not covered by any CBS. 

Fig.~\ref{fig:radius-result}  illustrates how the routing decisions of each of the three policies affect the distribution of load shares among all CBSs. It displays the average maximum  (upper curve) and minimum load share (lower curve)   over the mean number of CBSs a covered user can see. 



\begin{figure}[h]
\centering
\includegraphics[width=\columnwidth]{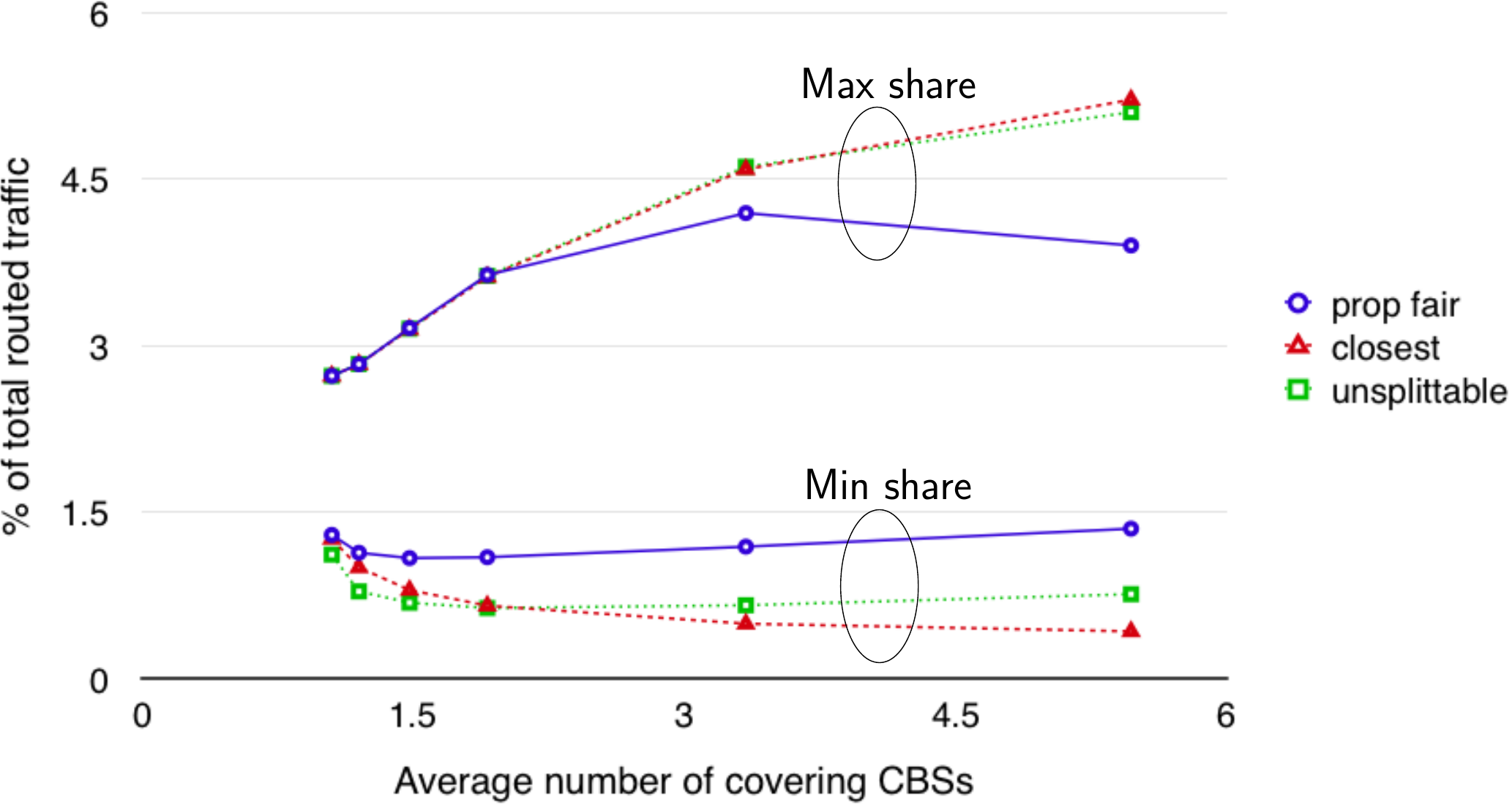}
\caption{Minimum and maximum load share of a CBS in the network depending on the mean coverage number.}
\label{fig:radius-result}
\end{figure}

The results show that with an increasing average number of covering stations, the fair policy achieves a lower maximum load as well as a higher minumum load. Since the overall traffic routed to the CBSs is the same for all policies, we can conclude that the  fair policy makes the most balanced use of the available resouces. The resources which remain available for potential cache-unrelated traffic are spread evenly across the network.

\subsection{Two-tier Network}

In a second scenario, we simulate an area covered by two tiers~\cite{DiRenzo2013}: one of large and one of small coverage.  We show that the fair policy is better at offloading traffic from larger to smaller CBSs than the other policies (see Fig.~\ref{fig:small-result}). While the fair policy burdens the small stations with a higher load, we  demonstrate that it distributes  the load   more evenly among them so that no individual station is overburdened (see Fig.~\ref{fig:small-minmax}).
The first tier consists of  large CBSs having a 187.5 m coverage radius while the second tier of small CBSs has 62.5 m coverage radius. The large CBSs are equipped with caches and  the two most popular files are stored  in all of them. The smaller CBSs get one of these two files assigned uniformly randomly in their cache.  

In Fig.~\ref{fig:small-result}, we show the percentage of all traffic  routed to large CBSs for each of the three policies. The x-coordinate increases with the ratio of  small stations over large stations in the network. When less traffic is routed to the large CBSs, then the policy provides a more efficient offloading of traffic towards the small CBSs. 
The figure shows that, increasing the amount of small CBSs, the fair policy offloads significantly more traffic to the small stations than  the unsplittable policy, and slightly more than the closest-available policy.


\begin{figure}[h]
\centering
\includegraphics[width=\columnwidth]{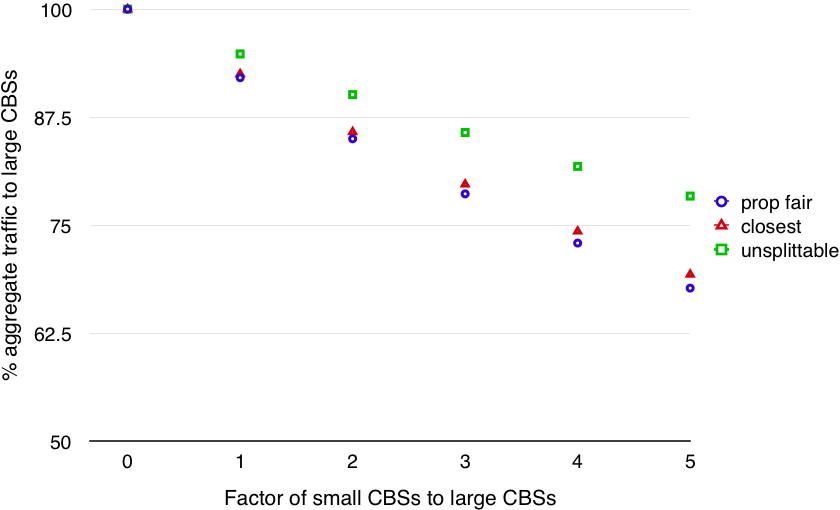}
\caption{Aggregate traffic share of large CBSs depending on the amount of small CBSs.}
\label{fig:small-result}
\end{figure}

Fig.~\ref{fig:small-minmax} shows (as in Fig.~\ref{fig:radius-result}) the maximum and minimum traffic load routed to a small CBS depending on each policy.  
Even though for the fair policy more users are routed to the small CBSs overall, the maximum load share that one  small CBS  takes is almost the same for all policies.    The increase in traffic load by the fair policy  is distributed to the less loaded small CBSs. This is indicated by the higher minimum load among the stations. When applying either the closest-available or the unsplittable policy these CBSs are underused.   Thus, the fair policy utilizes the available resources better.

\begin{figure}[h]
\centering
\includegraphics[width=\columnwidth]{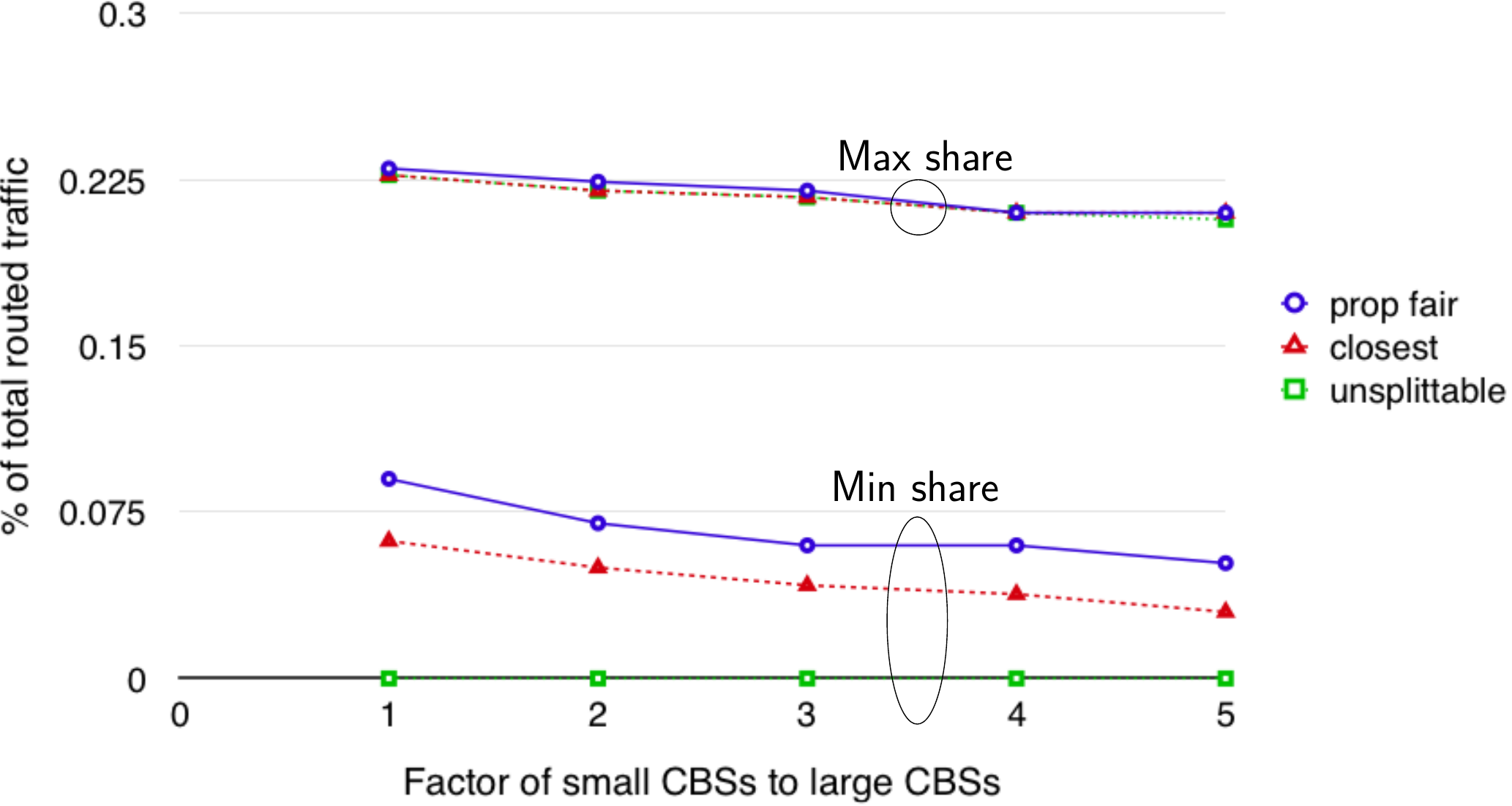}
\caption{Minimum and maximum load share of a small (2nd tier) CBS depending on the ratio of  small CBS number over large CBS number.}
\label{fig:small-minmax}
\end{figure}

\section{Conclusions}
\label{section:conclusion}

We propose a distributed algorithm that optimizes cache-related user traffic association among cache-equipped cellular  stations for a  given  content placement.  The main novelty of the approach is the balancing of cache-related traffic load to evenly use all available resources. The solution procedure makes use of the Augmented Lagrangian  to solve a strictly concave problem. The novel Bucket-filling subroutine takes optimal routing decisions locally at each  station.  With limited  information exchange,  our policy achieves a  fair distribution of users among the  stations.  The efficiency of the algorithm allows for large networks with different sizes of coverage areas and different cache sizes to be solved optimally. Simulations of single-tier and multi-tier networks show that our policy is superior to conventional user traffic association policies in balancing cache-related user traffic among stations.

\bibliographystyle{unsrt}
\bibliography{DocBib}

\begin{thebibliography}{10}

\bibitem{Cisco2016}
Cisco visual networking index: Global mobile data traffic forecast update,
  2015--2020 white paper.
\newblock White Paper, 2 2016.

\bibitem{poularakis2016}
K.~Poularakis, G.~Iosifidis, I.~Pefkianakis, Leandros
  Tassiulas, and M.~May.
\newblock Mobile data offloading through caching in residential 802.11 wireless
  networks.
\newblock {\em IEEE Transactions on Network and Service Management},
  13(1):71--84, 2016.

\bibitem{poularakis2016a}
K.~Poularakis, G.~Iosifidis, V.~Sourlas, and L.~Tassiulas.
\newblock Exploiting caching and multicast for 5g wireless networks.
\newblock {\em IEEE Transactions on Wireless Communications}, 15(4):2995--3007,
  April 2016.

\bibitem{Borst2010}
S.~Borst, V.~Gupta, and A.~Walid.
\newblock Distributed caching algorithms for content distribution networks.
\newblock In {\em Proceedings of the 29th Conference on Information
  Communications}, INFOCOM'10, pages 1478--1486, Piscataway, NJ, USA, 2010.
  IEEE Press.

\bibitem{Caire2013}
K.~Shanmugam, N.~Golrezaei, A.G. Dimakis, A.F. Molisch, and G.~Caire.
\newblock Femtocaching: Wireless content delivery through distributed caching
  helpers.
\newblock {\em Information Theory, IEEE Transactions on}, 59(12):8402--8413,
  Dec 2013.

\bibitem{Blaszczyszyn2014}
B.~Blaszczyszyn and A.~Giovanidis.
\newblock Optimal geographic caching in cellular networks.
\newblock In {\em Communications (ICC), 2015 IEEE International Conference on},
  pages 3358--3363. IEEE, 2015.

\bibitem{Bastug2015}
E.~Bastug, M.~Bennis, and M.~Debbah.
\newblock Cache-enabled small cell networks: Modeling and tradeoffs.
\newblock In {\em Wireless Communications Systems (ISWCS), 2014 11th
  International Symposium on}, pages 649--653, Aug 2014.

\bibitem{Poularakis2014b}
K.~Poularakis, G.~Iosifidis, and L.~Tassiulas.
\newblock Approximation algorithms for mobile data caching in small cell
  networks.
\newblock {\em IEEE Transactions on Communications}, 62(10):3665--3677, 2014.

\bibitem{Dehghan2015}
M.~Dehghan, A.~Seetharam, Bo~Jiang, Ting He, Th. Salonidis, J.~Kurose,
  D.~Towsley, and R.~Sitaraman.
\newblock On the complexity of optimal routing and content caching in
  heterogeneous networks.
\newblock In {\em INFOCOM, IEEE Conference on Computer Communications}, 2015.

\bibitem{Naveen2015}
K.P. Naveen, L.~Massoulie, E.~Baccelli, A.~Carneiro~Viana, and D.~Towsley.
\newblock On the interaction between content caching and request assignment in
  cellular cache networks.
\newblock In {\em Proceedings of the 5th Workshop on All Things Cellular:
  Operations, Applications and Challenges}, AllThingsCellular '15, pages
  37--42, New York, NY, USA, 2015. ACM.

\bibitem{Kelly1997}
F.~Kelly.
\newblock Charging and rate control for elastic traffic.
\newblock {\em European Transactions on Telecommunications}, 1997.

\bibitem{Mo2000}
J.~Mo and J.~Walrand.
\newblock Fair end-to-end window-based congestion control.
\newblock {\em IEEE/ACM Trans. Netw.}, 8(5):556--567, October 2000.

\bibitem{Bertsekas1989}
D.~P. Bertsekas and J.~N. Tsitsiklis.
\newblock {\em Parallel and Distributed Computation: Numerical Methods}.
\newblock Prentice-Hall, Inc., Upper Saddle River, NJ, USA, 1989.

\bibitem{Ruszczynski1995}
A.~Ruszczy{\'n}ski.
\newblock On convergence of an augmented lagrangian decomposition method for
  sparse convex optimization.
\newblock {\em Mathematics of Operations Research}, 20(3):634--656, 1995.

\bibitem{DiRenzo2013}
M.~Di~Renzo, A.~Guidotti, and G.~E. Corazza.
\newblock Average rate of downlink heterogeneous cellular networks over
  generalized fading channels: A stochastic geometry approach.
\newblock {\em IEEE Trans.\ on Comm.}, 61(7):3050--3071, 2013.

\end{thebibliography}
 
\end{document}